\documentclass{amsart}
\usepackage{mathrsfs,amssymb}

\theoremstyle{plain}
\newtheorem{theorem}{Theorem}[section]
\newtheorem{lemma}[theorem]{Lemma}
\newtheorem{proposition}[theorem]{Proposition}
\newtheorem{corollary}[theorem]{Corollary}

\theoremstyle{definition}
\newtheorem{notation}[theorem]{Notation}
\newtheorem{example}[theorem]{Example}
\newtheorem{definition}[theorem]{Definition}

\theoremstyle{remark}
\newtheorem{remark}[theorem]{Remark}

\input xy
\xyoption{all}

\def\au{\mathcal{A}}

\begin{document}
\title[Hanna Neumann Property in Free Monoids]{A Sufficient Condition for Hanna Neumann Property of Submonoids of a Free Monoid}
\author[S. N. Singh, K. V. Krishna]{Shubh Narayan Singh and  K. V. Krishna}
\address{Department of Mathematics, Indian Institute of Technology Guwahati, Guwahati, India}
\email{\{shubh, kvk\}@iitg.ac.in}


\begin{abstract}
Using automata-theoretic approach, Giambruno and Restivo have investigated on the intersection of two finitely generated submonoids of the free monoid over a finite alphabet. In particular, they have obtained Hanna Neumann property for a special class of submonoids generated by finite prefix sets. This work continues their work and provides a sufficient condition for Hanna Neumann property for the entire class of submonoids generated by finite prefix sets. In this connection, a general rank formula for the submonoids which are accepted by semi-flower automata is also obtained.
\end{abstract}



\maketitle

\section*{Introduction}

Howson proved that the intersection of two finitely generated subgroups of a free group  is finitely generated \cite{howson54}.
In 1956, Hanna Neumann improved the result that if $H$ and $K$ are finite rank subgroups of a free group, then
\[\widetilde{rk}(H \cap K) \le 2\widetilde{rk}(H)\widetilde{rk}(K),\] where $\widetilde{rk}(N) = \max(0, rk(N)-1)$ for a subgroup $N$ of rank $rk(N)$. Further, Neumann conjectured that\\

\hfill $\widetilde{rk}(H \cap K) \le \widetilde{rk}(H)\widetilde{rk}(K),$ \hfill ($\star$)\\

\noindent which is known as Hanna Neumann conjecture \cite{neu56}. In 1990, Walter Neumann proposed a stronger form of the conjecture called strengthened Hanna Neumann conjecture (SHNC) \cite{neu90}. Meakin and Weil proved SHNC for the class of positively generated subgroups of a free group \cite{meakin02}. The conjecture has recently been settled by Mineyev (cf. \cite{min11,min12}) and announced independently by Friedman (cf. \cite{fried-note,fried-arxiv}).

In contrast, it is not always true that the intersection of two finitely generated submonoids of a free monoid is finitely generated. It appears that the intersection problem for submonoids of free monoids is much more complex than the analogous problem for subgroups of free groups. In particular,  Hanna Neumann property for  submonoids of a free monoid is of special interest. Two finitely generated submonoids $H$ and $K$ of a free monoid are said to satisfy  \emph{Hanna Neumann property} (in short, HNP), if $H$ and $K$ satisfy the inequality ($\star$). There are several contributions in the literature to study the intersection of two submonoids of a free monoid.

In 1972, Tilson proved that the intersection of free submonoids of the free monoid over a finite alphabet is free \cite{til72}. In connection to HNP, Karhum\"aki obtained a result for submonoids of rank two of the free monoid over a finite alphabet. In fact, Karhum\"aki proved that the intersection of two submonoids of rank two is  generated either by a set of at most two words or by a regular language of a special form \cite{karhumaki84}. Using automata-theoretic approach, Giambruno and Restivo have investigated on the rank and HNP of certain submonoids of a free monoid \cite{giam08}.

In \cite{giam08}, Giambruno and Restivo introduced the concept called semi-flower automata (in short, SFA). An SFA accepts a finitely generated submonoid of the free monoid over the underlying alphabet, and vice versa. Moreover, if an SFA is deterministic, it accepts the submonoid generated by a finite prefix set. Conversely, the submonoid generated by a finite prefix set is accepted by a deterministic SFA with at most one `branch point going in' (in short, bpi). Using SFA, Giambruno and Restivo  have initiated the investigations on the intersection of two submonoids generated by finite prefix sets. It is clear that the product automaton of two deterministic SFA accepts the intersection of the submonoids accepted by the deterministic SFA.  If the product automaton is also semi-flower, then clearly the intersection is also finitely generated.  Giambruno and Restivo have considered two submonoids which are accepted by deterministic SFA with a unique bpi such that the product automaton is also semi-flower with at most one bpi. They have shown that such submonoids satisfy HNP. Further, if the product automaton has more than one bpi, they have provided some examples of submonoids which fail to satisfy HNP. Recently in \cite{singh11a,singh11}, Singh and Krishna have obtained a condition on the product automaton with two bpi's, so that the respective submonoids satisfy HNP.

The present work is in the direction of addressing HNP for the entire class of submonoids generated by finite prefix sets. This work generalizes the work of Singh and Krishna and provides a sufficient condition for HNP of two submonoids generated by finite prefix sets. Further, a general rank formula for the submonoids which are accepted by SFA is also obtained. The main work of the paper is presented in Section 2. Before that, in Section 1, the necessary preliminary concepts and results are presented. The paper is concluded in Section 3.

\section{Preliminaries}

In this section, we present the necessary background material from  \cite{bers85,giam07,giam08}. Let $A$ be a finite set called an \emph{alphabet} with its elements as \emph{letters}. The free monoid over $A$ is denoted by $A^*$ and $\varepsilon$ denotes the empty word -- the identity element of $A^*$. It is known that every submonoid of $A^*$ is generated by a unique minimal set of generators. Thus, if $H$ is a submonoid of $A^*$, then the \emph{rank} of $H$, denoted by $rk(H)$, is defined as the cardinality of the minimal set of generators $X$ of $H$, i.e. $rk(H) = |X|$. Further, the \emph{reduced rank} of $H$ is defined as $\max(0, rk(H)-1)$ and it is denoted by $\widetilde{rk}(H)$. A set of words, say $X$, is said to be a \emph{prefix set} if no element of $X$ is a proper prefix of another element in $X$.

Let $A$ be an alphabet. An \emph{automaton} over $A$ is a quadruple $(Q, I, T, \mathcal{F})$, where $Q$ is a finite set called the set of \emph{states}, $I$ and $T$ are subsets of $Q$ called the sets of \emph{initial} and \emph{final} states, respectively, and $\mathcal{F}\subseteq Q\times A\times Q$ called the set of \emph{transitions}. Clearly, by denoting the states as vertices/nodes and the transitions as labeled arcs, an automaton can be represented by a digraph in which initial and final states shall be distinguished appropriately.

In what follows, let $\au = (Q, I, T, \mathcal{F})$ be an automaton over $A$. A state, say $q$, of $\au$ is called a \emph{branch point going in}, in short \emph{bpi}, if the number of transitions coming into $q$ (i.e. the indegree of $q$ -- the number of arcs coming into $q$ -- in the digraph of $\au$) is at least two. In this work, we write $BPI(\au)$ to denote the set of bpi's of $\au$.  Further, for $i \ge 0$, we write
\[BPO_i(\au) = \{q \in Q\; |\; \mbox{ the number of transitions defined on $q$ is equal to }  i\},\] i.e. the set of states whose outdegree -- the number of arcs going out of the state -- in the digraph of $\au$ is $i$.

A \emph{path} in $\au$ is a finite sequence of consecutive arcs in its digraph. For $p_i \in Q$ ($0\le i \le k$) and $a_j \in A$ ($1 \le j \le k$), let
\[p_0 \xrightarrow{a_1} p_1 \xrightarrow{a_2} p_2 \xrightarrow{a_3} \cdots \xrightarrow{a_{k-1}} p_{k-1} \xrightarrow{a_k} p_k\] be a path, say $P$, in $\au$ that is starting at $p_0$ and ending at $p_k$. In this case, we write $i(P) =p_0$ and $f(P) = p_k$.  The word $a_1\cdots a_k \in A^*$ is the \emph{label of the path} $P$. A \emph{null path} is a path from a state to itself labeled by $\varepsilon$.

A path in $\au$ is called \emph{simple} if all the states on the path are distinct. A path that starts and ends at the same state is called as a \emph{cycle}, if it is not a null path. A cycle with all its states are distinct is called a \emph{simple cycle}. Other notions related to paths, viz. subpath, prefix path and suffix path, can be defined in a usual way or one may refer to \cite{giam08}.

The \emph{language accepted by $\au$}, denoted by $L(\au)$, is the set of words in $A^*$ that are the labels of the paths from an initial state to a final state. A state, say $q$, of $\au$ is \emph{accessible} (respectively, \emph{coaccessible}) if there is a path from an initial state to $q$ (respectively, a path from $q$ to a final state). The \emph{trim part} of $\au$, denoted by $\au^T$, is the automaton obtained from $\au$ by considering only the accessible and coaccessible states, and the respective transitions between them. Note that $L(\au) = L(\au^T)$. If $\au = \au^T$, then we say $\au$ is \emph{trim}. An automaton is \emph{deterministic} if it has a unique initial state and there is at most one transition defined for a state and a letter, i.e. the transition relation can be seen as a partial function from $Q \times A$ to $Q$.

An automaton is called a \emph{monoidal automaton} if it is trim with a unique initial state that is equal to a unique final state. Further, a monoidal automaton is called a \emph{semi-flower automaton}, in short SFA, if all the cycles in the automaton visit the unique initial-final state.

If $\au = (Q, I, T, \mathcal{F})$ is a monoidal automaton, we denote the initial-final state by $q_0$. In which case, we simply write $\au = (Q, q_0, q_0, \mathcal{F})$.  Further, if $\au$ is an SFA, let us denote by $C_{\au}$ the set of simple cycles (passing through $q_0$) in $\au$ and by $Y_{\au}$ the set of their labels.

In the following theorem we state the correspondence between SFA and finitely generated submonoids of a free monoid.

\begin{theorem}[\cite{giam08}]\label{thm1}
If $\au$ is an SFA over $A$, then $Y_{\au}$ is finite and $\au$ accepts the submonoid generated by $Y_{\au}$ in $A^*$. Moreover, if $\au$ is deterministic, then $Y_{\au}$ is a prefix set and it is the minimal set of generators of the submonoid accepted by $\au$. Conversely, let $X$ be a finite subset of $A^*$ and let $H$ be the submonoid generated by $X$; then there exists an SFA accepting $H$. Furthermore, if $X$ is a prefix set, then there exists a deterministic SFA with at most one bpi accepting $H$.
\end{theorem}

We now present the concept of product automaton and discuss the state of the art of the intersection problem of two submonoids generated by finite prefix sets.

Let $\au_1 = (Q_1, I_1, T_1, \mathcal{F}_1)$ and $\au_2 = (Q_2, I_2, T_2, \mathcal{F}_2)$ be two automata defined over $A$. The \emph{product automaton}  $\au_1  \times \au_2$ is the automaton $$(Q_1 \times Q_2, I_1 \times I_2, T_1 \times T_2, \mathcal{F})$$ over the alphabet $A$ such that
\[((p, p'), a, (q, q')) \in \mathcal{F} \Longleftrightarrow (p, a, q) \in \mathcal{F}_1 \; \mbox{ and } (p', a, q' ) \in \mathcal{F}_2\] for all $p, q \in Q_1$, $p', q' \in Q_2$ and $a \in A$.

Notice that, if $\au_1$ and $\au_2$ are deterministic, then so is $\au_1 \times \au_2$. But, if $\au_1$ and $\au_2$ are trim, then $\au_1 \times \au_2$ need not be trim. Further,
\[L(\au_1 \times \au_2) = L((\au_1 \times \au_2)^T) = L(\au_1) \cap L(\au_2).\]

Let $H$ and $K$ be submonoids generated by finite prefix sets of words over $A$. In view of Theorem \ref{thm1}, suppose $\au_H$ and $\au_K$ are  deterministic SFA over $A$ with at most one bpi accepting $H$ and $K$, respectively. Clearly, $(\au_H \times \au_K)^T$ is a deterministic monoidal automaton accepting $H \cap K$. In order to consider the case that $H \cap K$ is finitely generated, one could restrict $(\au_H \times \au_K)^T$  to be semi-flower. With this hypothesis, we discuss on HNP of $H$ and $K$ as follows.

\begin{description}
  \item[Case 1. $\au_H$ or $\au_K$ has no bpi's] In this case, $(\au_H \times \au_K)^T$ has no bpi's. Consequently, $H \cap K$ is cyclic, so that $H$ and $K$ satisfy HNP.
  \item[Case 2. $\au_H$ and $\au_K$ have unique bpi] In this case, $(\au_H \times \au_K)^T$ can have arbitrary number of bpi's. Thus, the problem is considered into various subcases and we only know the following.
      \begin{description}
      \item[I. $(\au_H \times \au_K)^T$ has at most one bpi] In this subcase, if $(\au_H \times \au_K)^T$ is semi-flower, then $H$ and $K$ satisfy HNP (cf. \cite[Theorem 3.6]{giam08}).
      \item[II. $(\au_H \times \au_K)^T$ has two bpi's] In this subcase, even if $(\au_H \times \au_K)^T$ is semi-flower, $H$ and $K$ need not satisfy HNP (cf. \cite[Example 3.7]{singh11}). However, if $(\au_H  \times \au_K)^T$ is an SFA with two bpi's having a unique path from one bpi to the other, then $H$ and $K$ satisfy HNP (cf. \cite[Corollary 3.11]{singh11}).
      \end{description}
\end{description}
In general, if $(\au_H  \times \au_K)^T$ has more than one bpi, there are several examples of $H$ and $K$ which fail to satisfy HNP (cf. \cite{giam07,giam08}). Thus, if $(\au_H  \times \au_K)^T$ has arbitrary number of bpi's, in this work we would investigate on certain conditions so that $H$ and $K$ satisfy HNP. For that purpose, we would require the following supplementary results from \cite{giam08}.

\begin{proposition}\label{result-1}
Let $A$ be an alphabet of cardinality $n$. If $\au = (Q, q_0, q_0, \mathcal{F})$ is a deterministic SFA over $A$, then
\[|\mathcal{F}| - |Q| = \displaystyle\sum_{i = 2}^n |BPO_i(\au)|(i - 1).\]
\end{proposition}

\begin{proposition}\label{result-2}
Let $A$ be an alphabet of cardinality $n$ and let $\au_1$ and $\au_2$ be two deterministic automata over $A$. If
$c_i = |BPO_i(\au_1)|$ and $d_i = |BPO_i(\au_2)|$, for each $i \in \{1, \ldots, n\}$, then
\[|BPO_t(\au_1\times \au_2)|\leq\displaystyle\sum_{t\leq r,s\leq n}c_r d_s.\]
\end{proposition}

\begin{proposition}\label{result-3}
Let  $\langle c_1,\ldots,c_n \rangle$ and $\langle d_1,\ldots,d_n \rangle$ be two finite sequences of natural numbers; then
\[\sum_{t = 2}^n(t-1)\left(\sum_{t\leq r\leq n}c_r \sum_{t\leq s\leq n} d_s \right) \leq
\left(\sum_{i = 2}^n (i-1)c_i\right)\left(\sum_{j = 2}^n (j-1)d_j\right).\]
\end{proposition}

\section{Main Results}

In this section we generalize the work of Singh and Krishna in \cite{singh11}. First we obtain the rank of the submonoid of a free monoid that is accepted by an SFA. Then we proceed to obtain a condition for HNP of two submonoids of a free monoid which are generated by finite prefix sets.

\subsection{BPR and Rank}

We begin with introducing a concise notation for an SFA in which only the initial-final state, bpi's and the respective paths between them will be considered along with their labels. We call this as \emph{\underline{b}pi's and \underline{p}aths \underline{r}epresentation}, in short BPR, of the semi-flower automaton.

\begin{definition}
Let $\au = (Q, q_0, q_0, \mathcal{F})$ be an SFA over $A$; the BPR of $\au$ is a quadruple $\au' = (Q', q_0, q_0, \mathcal{F}')$, where
\begin{enumerate}
\item[(i)] $Q' = BPI(\au) \cup \{q_0\}$, and

\item[(ii)] $\mathcal{F}'$ is the finite subset of $Q' \times A^* \times Q'$ defined by
$(p_0 = p, x, q = p_k) \in \mathcal{F}'$ if and only if there exist distinct  $p_1, \ldots, p_{k-1} \in Q \setminus Q'$ and $x = a_1\cdots a_k$, for $a_i \in A$, such that $(p_{i-1}, a_i, p_i) \in \mathcal{F}$ for all $1 \le i \le k$, i.e.
\[p = p_0 \xrightarrow{a_1} p_1 \xrightarrow{a_2} p_2 \xrightarrow{a_3} \cdots \xrightarrow{a_{k-1}} p_{k-1} \xrightarrow{a_k} p_k = q\] is a simple path from $p$ to $q$ (or simple cycle, when $p = q$) labeled by $x$ in which the intermediate nodes, if any, are outside $Q'$.
\end{enumerate}
\end{definition}

\begin{remark}
By adopting the digraph representation of an automaton, we can draw a digraph for the BPR of an SFA. Here, the arcs are labeled by the labels (words) of respective simple paths (or simple cycles) of the SFA.
\end{remark}

\begin{example}
The BPR of the SFA given in \textsc{Figure} \ref{fig3} is shown in \textsc{Figure} \ref{fig4}. Here, we distinguish the initial-final state $q_0$ by double circles.
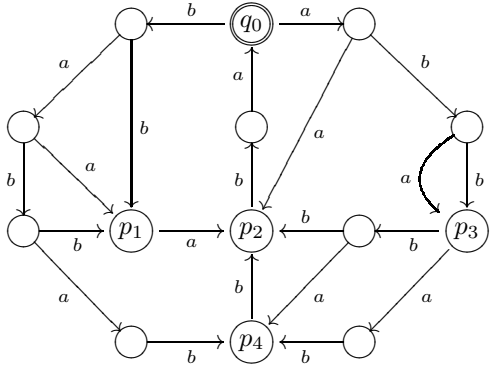
\begin{figure}[t]
\entrymodifiers={++[o][F-]} \SelectTips{cm}{}
\[\xymatrix{*\txt{} & \ar[ld]_a \ar[dd]^b & *++[o][F=]{q_0} \ar[r]^a \ar[l]_b  & \ar[rd]^b \ar[ddl]^a \\
\ar[d]_b \ar[dr]^a & *\txt{} & \ar[u]^a & *\txt{} & \ar[d]^b \ar@/_1.5pc/[d]_a \\
\ar[r]_b \ar[rd]_a & p_1 \ar[r]_a & p_2 \ar[u]^b & \ar[l]_b \ar[ld]^a & \ar[dl]^a \ar[l]^b p_3 \\
*\txt{} & \ar[r]_b & \ar[u]^b p_4  & \ar[l]^b \\
}\]
\caption{A Semi-Flower Automaton}
\label{fig3}
\end{figure}

\begin{figure}[t]
\entrymodifiers={++[o][F-]} \SelectTips{cm}{}
\[\xymatrix{ *\txt{} & *\txt{} &  *++[o][F=]{q_0} \ar[rrdd]^{aba} \ar@/^2pc/[rrdd]^{abb} \ar[dddr]^{babab} \ar[ddll]_{bb} \ar@/_2pc/[ddll]_{baa} \ar@/_4pc/[ddll]_{babb} \ar@/^1.5pc/[dddl]^{aa}  \\
*\txt{}\\
p_1\ar[rd]_a & *\txt{} &*\txt{} &*\txt{} & p_3 \ar[ld]^{ab} \ar@/_1.5pc/[ld]_{ba} \ar@/^5pc/[dlll]_{bb}\\
*\txt{} & p_2 \ar[uuur]^{ba} &*\txt{} & p_4 \ar[ll]^b }\]
\caption{The BPR of the SFA given in \textsc{Figure} \ref{fig3}}
\label{fig4}
\end{figure}
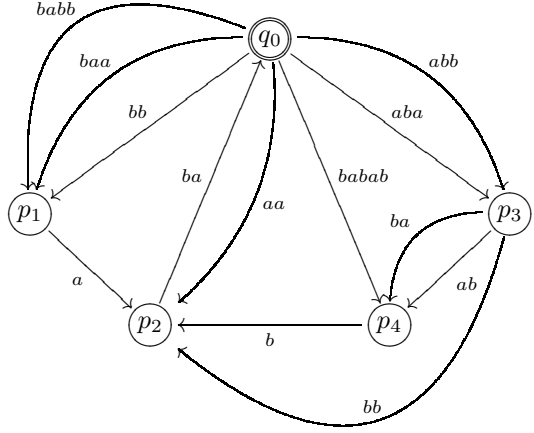
\end{example}

\begin{remark}\label{cyc-bpr} Let $\au = (Q, q_0, q_0, \mathcal{F})$ be an SFA and let $\au' = (Q', q_0, q_0, \mathcal{F}')$ be the BPR of $\au$.
\begin{enumerate}
\item[(i)] Every cycle in $\au'$ passes through the state $q_0$.
\item[(ii)] The number of simple cycles in $\au$ is equal to the number of simple cycles in $\au'$.
\item[(iii)] For any $p, q \in Q'$, the number simple paths from $p$ to $q$ in $\au$ is equal to the number of simple paths from $p$ to $q$ in $\au'$.
\end{enumerate}
\end{remark}

\begin{proposition}\label{top-ord}
Let $\au$ be an SFA and let $\au'$ be the BPR of $\au$. There is a linear ordering $\preccurlyeq$  on the states of $\au'$ such that
\begin{enumerate}
\item[(i)] the initial-final state $q_0$ is the least element, and
\item[(ii)] for $j \ne 0$, if $q_j \preccurlyeq q_i$, then there is no arc from $q_j$ to $q_i$ in $\au'$.
\end{enumerate}
\end{proposition}

\begin{proof}
Construct the digraph $\mathcal{G}$ from the digraph of $\au'$ by removing the arcs which are leaving out of the initial-final state $q_0$. By Remark \ref{cyc-bpr}(i), $\mathcal{G}$ is a directed acyclic graph. Define a relation $\le$ on the nodes of $\mathcal{G}$ by \[p \le q\; \mbox{if and only if there is a simple path from $q$ to $p$ in $\mathcal{G}$.}\]
As the null path is a simple path from a node to itself in $\mathcal{G}$, clearly $\le$ is reflexive. Since there are no cycles in $\mathcal{G}$, it can be observed that  $\le$ is anti-symmetric and transitive. Thus, the relation $\le$ is a partial ordering on the nodes of $\mathcal{G}$. Since every state in $\au$ is coaccessible, the initial-final state $q_0$ is the least element with respect to $\le$.

As every partial ordering can be extended to a linear ordering, consider a linear ordering $\preccurlyeq$ of the nodes of $\mathcal{G}$ which is an extension of $\le$.  Thus, the obtained linear ordering $\preccurlyeq$ is the desired one.
\end{proof}

\begin{remark}
By applying a topological sort algorithm (eg. refer \cite{cor01}) on the directed acyclic graph $\mathcal{G}$, one can get a linear ordering as described in Proposition \ref{top-ord}.
\end{remark}

In what follows, by a \emph{topological ordering} of the bpi's of an SFA is meant a linear ordering on the states (possibly, except the initial-final state) of its BPR as in Proposition \ref{top-ord}.

\begin{example}
Let $\au$ be the SFA given in \textsc{Figure} \ref{fig3}. A topological ordering of the bpi's of $\au$ is \[p_2, p_4, p_3, p_1.\] Notice that there is no arc from $p_3$ or $p_4$ to $p_1$ in the BPR of $\au$ (cf. \textsc{Figure} \ref{fig4}). Hence, in a topological ordering of the bpi's of $\au$, the bpi $p_1$ can come at any position after $p_2$. Thus, the possible other topological orderings are $p_2, p_1, p_4, p_3$ and $p_2, p_4, p_1, p_3$.
\end{example}

\begin{notation}
Let $\au$ be an SFA and let $\au'$ be the BPR of $\au$. If we say that $\au$ has $m$ bpi's, we always assume that $q_1, q_2, \ldots, q_m$ are the bpi's of $\au$, which are considered in a topological ordering. That is, $$q_1 \preccurlyeq q_2 \preccurlyeq \cdots \preccurlyeq q_m.$$  As per the ordering, we also fix the following numbers in the BPR $\au'$.
\begin{enumerate}
\item[(i)] For $1 \le i \le m$, $\kappa_i$ always refer to the number of arcs from the state $q_0$ to the bpi $q_i$ in $\au'$.
\item[(ii)] For $1 \le i, j \le m$, $\kappa_{ij}$ always refer to the number of arcs from the bpi $q_i$ to the bpi $q_j$ in $\au'$.
\end{enumerate}
\end{notation}

\begin{remark}
As per the topological ordering of the bpi's of $\au$, we have $\kappa_{ij} = 0$ for all $j \ge i > 1$. If the initial-final state $q_0$ of $\au$ is a bpi, then clearly $q_1 = q_0$ so that, for $j \ge 1$, $\kappa_{1j} = \kappa_j$; otherwise, $\kappa_{1j} = 0$.
\end{remark}

\begin{remark}\label{indeg-bpi}
In the digraph of $\au$ (as well as in $\au'$), the indegree of a bpi $q_j$, $1 \le j \le m$, is given by the expression \[ \kappa_j + \displaystyle\sum_{i = j + 1}^m \kappa_{ij}.\]
\end{remark}

The following lemma is useful in obtaining the rank of an SFA.

\begin{lemma}\label{lem-rk}
Let $\au$ be an SFA and let $p$ be the first bpi in a topological ordering of the bpi's of $\au$, i.e. $p \preccurlyeq q$, for all bpi's $q$; then
\begin{enumerate}
\item[\rm(i)] there is a unique simple path from $p$ to the initial-final state $q_0$, and
\item[\rm(ii)] every cycle in $\au$ visits $p$.
\end{enumerate}
\end{lemma}

\begin{proof} If the initial-final state $q_0$ is a bpi, then clearly $p = q_0$. In which case, the null path is the unique simple path from $p$ to $q_0$. And, since $\au$ is semi-flower, every cycle in $\au$ visits $p$. If $q_0$ is not a bpi, we proceed as follows.

\begin{enumerate}
\item[(i)] Since $p$ is coaccessible, there is a path from $p$ to the state $q_0$. Now suppose there are two different paths $P_1$ and $P_2$ with labels $u$ and $v$, respectively, from $p$ to the state $q_0$. Let $w$ be the label of longest suffix path $P'$ which is in common between the paths $P_1$ and $P_2$. As the state $q_0$ is not a bpi, $w \neq \varepsilon$. But then $i(P')$ will be a bpi different from $p$. This a contradiction to the choice of $p$. Thus, there is a unique simple path from $p$ to $q_0$.

\item[(ii)] Suppose there is a cycle that is not visiting $p$. Then the cycle contributes one to the indegree of the state $q_0$. Also, from above (i), there is a path from $p$ to the state $q_0$. This implies that the state $q_0$ is a bpi; a contradiction.
\end{enumerate}
\end{proof}

\begin{corollary}
Let $\au$ be an SFA. If $p$ is the first bpi in a topological ordering of the bpi's of $\au$, then $p$ is the first bpi in any topological ordering of the bpi's of $\au$.
\end{corollary}

Let $\au$ be an SFA. Now we are ready to present our first main result of the paper on the rank of the submonoid $L(\au)$. The rank of $L(\au)$ can be characterized using the bpi's of $\au$. First note that, if there is no bpi in $\au$, then clearly the rank of $L(\au)$ is either $0$ or $1$. If $\au$ has at least one bpi, we have the following theorem.

\begin{theorem}\label{grk}
Let $\au$ be an SFA and $m \ge 1$. If $\au$ has $m$ bpi's, then\\

\hfill$rk(L(\au)) \leq \displaystyle\sum_{i=1}^m \kappa_i \overline{\kappa_{i0}},$\hfill {\rm(\#)}\\

\noindent where $\overline{\kappa_{i0}}$ is the number of simple paths from the bpi $q_i$ to the initial-final state $q_0$. The number $\overline{\kappa_{i0}}$ can be given by the recursive formula
\[\overline{\kappa_{10}} = 1\;  \mbox{ and }\;   \overline{\kappa_{i0}} = \sum_{j=1}^{i-1} \kappa_{ij} \overline{\kappa_{j0}},\; \mbox{ for }i > 1.\]
Moreover, if $\au$ is deterministic, then the equality holds in {\rm(\#)}.
\end{theorem}

\begin{proof}
Let $q_1 \preccurlyeq q_2 \preccurlyeq \cdots \preccurlyeq q_m$ be the bpi's of $\au$. It is known from Theorem \ref{thm1} that \[rk(L(\au)) \le |Y_{\au}| \le |C_\au|.\] We prove the result by showing that $|C_\au|$, the number of simple cycles in $\au$ passing through the state $q_0$, is equal to the righthand side of (\#), i.e. we show that
\[|C_\au| = \sum_{i=1}^m \kappa_i \overline{\kappa_{i0}}.\]
By Remark \ref{cyc-bpr}(ii), $|C_\au| = |C_{\au'}|$, where $C_{\au'}$ is the number of simple cycles in the BPR $\au'$ of $\au$.

For $1 \le i \le m$, let $\nu_i$ be the number of simple cycles in $\au'$ that are passing through the bpi $q_i$  but not through any bpi $q_j$ with $j > i$. Clearly, \[|C_{\au'}| = \sum_{i = 1}^m \nu_i.\] We conclude the result by arguing that $\nu_i = \kappa_i \overline{\kappa_{i0}}$, for $1 \le i \le m$.

In case $i = 1$, $\nu_1$ is the number of simple cycles in $\au'$ that are passing through the bpi $q_1$ but not through any other bpi. First note that, by Lemma \ref{lem-rk} there is unique simple path from $q_1$ to $q_0$ so that $\overline{\kappa_{10}} = 1$. Now, each simple cycle that is counted in $\nu_1$ is merely an arc from $q_0$ to $q_1$ followed by the unique simple path from $q_1$ to $q_0$.  Thus, the number of simple cycles counted in $\nu_1$ is the number of arcs from $q_0$ to $q_1$, i.e. $\kappa_1$. Hence, we have
\[\nu_1 = \kappa_1 = \kappa_1\overline{\kappa_{10}}.\]

For $i > 1$, as per the topological ordering, $\nu_i$  is clearly obtained by multiplying the number of arcs from $q_0$ to the bpi $q_i$ and the number of simple paths from $q_i$ to $q_0$. That is,
\[\nu_i = \kappa_i\overline{\kappa_{i0}}\] as desired. Now, we obtain the recursive formula for $\overline{\kappa_{i0}}$. For $1 \le t < i$, let $\mu_{it}$ be the number of simple paths in $\au'$ from the bpi $q_i$ to $q_0$ that are passing through the bpi $q_t$ but not through any other bpi $q_j$ with $j > t$. Clearly, $\overline{\kappa_{i0}} = \displaystyle\sum_{t = 1}^{i-1} \mu_{it}$. But, for $1 \le t < i$, the number $\mu_{it}$ is nothing else but the product of the number of arcs from $q_{i}$ to $q_t$ and the number of simple paths from $q_t$ to $q_0$, i.e. $\mu_{it} = \kappa_{it}\overline{\kappa_{t0}}$. Hence, we have the recursive formula
\[\overline{\kappa_{i0}} = \sum_{t = 1}^{i-1}\kappa_{it}\overline{\kappa_{t0}}.\]

If $\au$ is deterministic, then by Theorem \ref{thm1}, we have
\[rk(L(\au)) = |C_\au| = \sum_{i=1}^m \kappa_i\overline{\kappa_{i0}}.\]
\end{proof}

Now, Theorem 2.10 of \cite{giam08} is an immediate corollary as stated below. We will use this corollary in one of our main results.

\begin{corollary}\label{th2.2}
If $\au = (Q, q_0, q_0, \mathcal{F})$ is an SFA with a unique bpi, then \[rk(L(\au)) \le \kappa_1 = |\mathcal{F}| - |Q| + 1.\]
Moreover, if $\au$ is deterministic, then the equality holds.
\end{corollary}

\begin{example}\label{2bpi-ctr}
Let us consider the topological ordering \[p_2 \preccurlyeq p_4 \preccurlyeq p_3 \preccurlyeq p_1\] of the bpi's of the SFA $\au$ given in \textsc{Figure} \ref{fig3}. That is, $q_1 = p_2$, $q_2 = p_4$, $q_3 = p_3$ and $q_4 = p_1$.  Accordingly, $\kappa_1 = 1, \kappa_2 = 1, \kappa_3 = 2$ and $\kappa_4 = 3$. Also,  $\kappa_{41} = 1, \kappa_{42} = 0, \kappa_{43} = 0$, $\kappa_{31} = 1, \kappa_{32} = 2$ and $\kappa_{21} =1$. Since $\au$ is deterministic, we have
\begin{eqnarray*}
rk(L(\au)) &=& \kappa_1 + \kappa_2(\kappa_{21}) + \kappa_3(\kappa_{31} + \kappa_{32}\kappa_{21}) +\\
&& \kappa_4(\kappa_{41} + \kappa_{42}\kappa_{21} + \kappa_{43}\kappa_{31}+ \kappa_{43}\kappa_{32}\kappa_{21})\\
&=& 11.
\end{eqnarray*}
\end{example}

\subsection{Hanna Neumann Property}

In this subsection, we obtain a sufficient condition for HNP of two submonoids which are accepted by deterministic SFA with a unique bpi. The following lemma is useful in obtaining the proposed result.

\begin{lemma}
Let $\au  = (Q, q_0, q_0,\mathcal{F})$ be an SFA and $m \ge 1$. If $\au$ has $m$ bpi's, then
\[|\mathcal{F}|-|Q|+1 \geq rk(L(\au)) - \sum_{i=2}^m \left((\kappa_{i}-1)(\kappa_{i1} - 1) + \sum_{j=2}^{i-1} \kappa_{ij}( \kappa_i \overline{\kappa_{j0}} - 1)\right).\]
Moreover, if $\au$ is deterministic, then the equality holds.
\end{lemma}

\begin{proof}
Since the number of transitions $|\mathcal{F}|$ of $\au$ is the total indegree (i.e. the sum of indegrees of all the states) of the digraph of $\au$, by Remark \ref{indeg-bpi},
we have
\[|\mathcal{F}|=|Q|-m + \sum_{j=1}^m \left( \kappa_j + \sum_{i=j+1}^m \kappa_{ij}\right).\]
Consequently,
\begin{eqnarray*}
|\mathcal{F}|-|Q|+1 &=& \kappa_1 + \sum_{j=2}^m \left(\kappa_{j}-1\right) +  \sum_{j=1}^m\sum_{i=j+1}^m \kappa_{ij}\\
&=& \kappa_1 + \sum_{i=2}^m \kappa_i \overline{\kappa_{i0}} + \sum_{j=2}^m \left(\kappa_{j}-1\right) +  \sum_{j=1}^m\sum_{i=j+1}^m \kappa_{ij} - \sum_{i=2}^m \kappa_i \overline{\kappa_{i0}}.
\end{eqnarray*}

Now, by Theorem \ref{grk} and simple algebraic manipulations, we have\\

\noindent $|\mathcal{F}|-|Q|+1$
\begin{eqnarray*}
&\geq & rk(L(\au)) + \sum_{j=2}^m \left(\kappa_{j}-1\right) +  \sum_{j=1}^m\sum_{i=j+1}^m \kappa_{ij} - \sum_{i=2}^m \kappa_i \overline{\kappa_{i0}}\\
&=& rk(L(\au)) + \sum_{j=2}^m \left(\kappa_{j}-1\right) +  \sum_{j=1}^m\sum_{i=j+1}^m \kappa_{ij} - \sum_{i=2}^m \kappa_i \left(\sum_{j=1}^{i-1} \kappa_{ij} \overline{\kappa_{j0}}\right)\\
&=& rk(L(\au)) + \sum_{j=2}^m \left(\kappa_{j}-1\right) +  \sum_{j=1}^m\sum_{i=j+1}^m \kappa_{ij} - \sum_{i=2}^m \kappa_i \left(\kappa_{i1} + \sum_{j=2}^{i-1} \kappa_{ij} \overline{\kappa_{j0}}\right)\\
&=& rk(L(\au)) + \sum_{j=2}^m \left(\kappa_{j}-1\right) +  \sum_{j=1}^m\sum_{i=j+1}^m \kappa_{ij} - \sum_{i=2}^m \kappa_i \kappa_{i1} - \sum_{i=2}^m \sum_{j=2}^{i-1} \kappa_i \kappa_{ij} \overline{\kappa_{j0}}\\
&=& rk(L(\au)) + \sum_{j=2}^m \left(\kappa_{j}-1\right) + \sum_{i=2}^m \kappa_{i1}+ \sum_{j=2}^m\sum_{i=j+1}^m \kappa_{ij} - \sum_{i=2}^m \kappa_i \kappa_{i1} - \sum_{i=2}^m \sum_{j=2}^{i-1} \kappa_i \kappa_{ij} \overline{\kappa_{j0}}\\
&=& rk(L(\au)) + \sum_{j=2}^m \left(\kappa_{j}-1\right) - \sum_{i=2}^m \kappa_{i1}(\kappa_i-1) + \sum_{j=2}^m\sum_{i=j+1}^m \kappa_{ij}  - \sum_{i=2}^m \sum_{j=2}^{i-1} \kappa_i \kappa_{ij} \overline{\kappa_{j0}}\\
&=& rk(L(\au)) - \sum_{i=2}^m \left(\kappa_{i}-1\right) (\kappa_{i1} - 1) + \sum_{j=2}^m\sum_{i=j+1}^m \kappa_{ij}  - \sum_{i=2}^m \sum_{j=2}^{i-1} \kappa_i \kappa_{ij} \overline{\kappa_{j0}}\\
&=& rk(L(\au)) - \sum_{i=2}^m \left(\kappa_{i}-1\right) (\kappa_{i1} - 1) + \sum_{j=2}^m\sum_{i=2}^m \kappa_{ij}  - \sum_{i=2}^m \sum_{j=2}^{m} \kappa_i \kappa_{ij} \overline{\kappa_{j0}},\\
&&\mbox{as $\kappa_{ij} = 0$ for all $j \ge i > 1$}\\
&=& rk(L(\au)) - \sum_{i=2}^m \left(\kappa_{i}-1\right) (\kappa_{i1} - 1) - \sum_{i=2}^m\sum_{j=2}^m \kappa_{ij}( \kappa_i \overline{\kappa_{j0}} - 1)\\
&=& rk(L(\au)) - \sum_{i=2}^m \left(\kappa_{i}-1\right) (\kappa_{i1} - 1) - \sum_{i=2}^m\sum_{j=2}^{i-1} \kappa_{ij}( \kappa_i \overline{\kappa_{j0}} - 1)\\
&&\mbox{as $\kappa_{ij} = 0$ for all $j \ge i > 1$}.
\end{eqnarray*}
Thus,
\[|\mathcal{F}|-|Q|+1 \geq rk(L(\au)) - \sum_{i=2}^m \left((\kappa_{i}-1)(\kappa_{i1} - 1) + \sum_{j=2}^{i-1} \kappa_{ij}( \kappa_i \overline{\kappa_{j0}} - 1)\right).\]
\end{proof}

Now, by Proposition \ref{result-1}, we have the following corollary.

\begin{corollary}\label{cor3.4}
Let $A$ be an alphabet of cardinality $n$ and let $\au$ be a deterministic SFA  over $A$. For $m \ge 1$, if $\au$ has $m$ bpi's, then
\begin{eqnarray*}
rk(L(\au)) &=& \sum_{i=2}^m \left(\left(\kappa_i-1\right)\left(\kappa_{i1}-1\right) +\sum_{j=2}^{i-1} \kappa_{ij}( \kappa_i \overline{\kappa_{j0}} - 1)\right)\\&&
+ \displaystyle\sum_{t = 2}^n |BPO_t(\au)|(t - 1) + 1.
\end{eqnarray*}
\end{corollary}

\begin{theorem}\label{ghn}
Let $\au_H$ and $\au_K$ be deterministic SFA over $A$ each with a unique bpi accepting submonoids $H$ and $K$, respectively. For $m \ge 1$, if the automaton $(\au_H  \times \au_K)^T$ is an SFA with $m$ bpi's, say $q_1, q_2, \ldots, q_m$ considered in a topological ordering, then
\begin{eqnarray*}
\widetilde{rk}(H \cap K) &\le & \sum_{i = 2}^m \left((\kappa_i - 1)(\kappa_{i1} - 1)  + \sum_{j = 2}^{i-1}
\kappa_{ij}\left(\kappa_i\overline{\kappa_{j0}}  - 1\right)\right) +  \widetilde{rk}(H)\widetilde{rk}(K),
\end{eqnarray*}
where $\kappa_i$ is the number of arcs from the initial-final state to $q_i$ and $\kappa_{ij}$ is the number of arcs from $q_i$ to $q_j$
in the BPR of $(\au_H \times \au_K)^T$.
\end{theorem}

\begin{proof}
Let $A$ be an alphabet of cardinality $n$. For $m \ge 1$,  note that
\begin{eqnarray*}
\widetilde{rk}(H\cap K) &=& rk(L(\au_H  \times \au_K))  - 1\\
&=& \sum_{i=2}^m \left(\left(\kappa_i-1\right)\left(\kappa_{i1}-1\right) +\sum_{j=2}^{i-1} \kappa_{ij}\left(\kappa_i \overline{\kappa_{j0}}   -1\right)\right)\\&&
+ \displaystyle\sum_{t = 2}^n |BPO_t(\au_H  \times \au_K)|(t - 1) \; \mbox{ by Corollary \ref{cor3.4}}\\
&\leq & \sum_{i=2}^m \left(\left(\kappa_i-1\right)\left(\kappa_{i1}-1\right) +\sum_{j=2}^{i-1} \kappa_{ij}\left(\kappa_i \overline{\kappa_{j0}} -1\right)\right)\\&&
+ \sum_{t = 2}^n (t - 1)\left(\sum_{t\leq r, s\leq n}c_rd_s \right)\; \mbox{ by Proposition \ref{result-2}},
\end{eqnarray*}
where $c_r = |BPO_r(\au_H)|$ and $d_s = |BPO_s(\au_K)|$. Further, by Proposition \ref{result-3}, we have

\begin{eqnarray*}
\widetilde{rk}(H\cap K) &\leq & \sum_{i=2}^m \left(\left(\kappa_i-1\right)\left(\kappa_{i1}-1\right) +\sum_{j=2}^{i-1} \kappa_{ij}\left(\kappa_i \overline{\kappa_{j0}} -1\right)\right)\\&&+ \left(\sum_{i = 2}^n (i - 1)c_i\right)\left(\sum_{j = 2}^n (j-1)d_j\right)\\
& = & \sum_{i=2}^m \left(\left(\kappa_i-1\right)\left(\kappa_{i1}-1\right) +\sum_{j=2}^{i-1} \kappa_{ij}\left(\kappa_i \overline{\kappa_{j0}} -1\right)\right) + \widetilde{rk}(H)\widetilde{rk}(K)
\\&& \mbox{by Corollary \ref{th2.2} and Proposition \ref{result-1}.}
\end{eqnarray*}
Hence the result.
\end{proof}

We now state a sufficient condition for Hanna Neumann property of the submonoids under consideration.
\begin{corollary}
In addition to the hypothesis of Theorem \ref{ghn}, if there is no path between any two bpi's $q_i$ and $q_j$, for $i, j > 1$, and there is a unique simple path from each bpi to $q_1$ in the automaton $(\au_H  \times \au_K)^T$, then \[\widetilde{rk}(H \cap K) \le \widetilde{rk}(H)\widetilde{rk}(K).\]
\end{corollary}

\begin{proof}
For $i, j > 1$, if there is no path between the bpi's $q_i$ and $q_j$, then $\kappa_{ij} = 0$. Further, for $i \geq 2$, if there is a unique simple path from each bpi $q_i$ to $q_1$, then the path cannot pass through any other bpi. Thus, we have $\kappa_{i1} = 1$ so that
\[\sum_{i=2}^m \left(\left(\kappa_i-1\right)\left(\kappa_{i1}-1\right) +\sum_{j=2}^{i-1} \kappa_{ij}\left(\kappa_i \overline{\kappa_{j0}} -1\right)\right)=0.\]
Hence, by Theorem \ref{ghn}, \[\widetilde{rk}(H \cap K) \le \widetilde{rk}(H)\widetilde{rk}(K).\]
\end{proof}

\section{Conclusion}

This work considers the intersection problem of two submonoids of a free monoid which are generated by finite prefix sets. In particular, this work has obtained a sufficient condition for Hanna Neumann property for the class of submonoids generated by finite prefix sets. In that connection, a general rank for the submonoids which are accepted by semi-flower automata is also obtained. Thus, this work addresses one of the problems, viz. the prefix case, posed by Giambruno and Restivo in the conclusions of the paper \cite{giam08}. However, there is a lot more to investigate on the general problem concerning the intersection of two arbitrary submonoids of a free monoid. For instance, even in the prefix case, one could investigate on the necessary and sufficient conditions for Hanna Neumann property.  On the other hand, the intersection problem of two submonoids generated by finite non-prefix sets of words is of particular interest. For this problem, the rank formula that is obtained (for nondeterministic automata) in this paper may be useful.

\section*{Acknowledgements}

The authors are very much thankful to anonymous referees for their valuable comments which improved the manuscript; particularly, for pointing out a mistake in Theorem 2.13.

\end{document}